\documentclass{llncs}

\usepackage{microtype}
\usepackage{color}
\usepackage[utf8]{inputenc}
\usepackage[english]{babel}
\usepackage{amsmath,amssymb}
\usepackage{subfig}
\usepackage{multirow}
\usepackage{tikz}
\usetikzlibrary{positioning}
\usetikzlibrary{shapes,snakes}
\usetikzlibrary{arrows,automata}
\usepackage{todonotes}
\usetikzlibrary{patterns}

\newcommand{\clocks}{\mathcal{C}}
\newcommand{\parameters}{\mathcal{P}}
\newcommand{\guards}{\mathcal{G}}
\newcommand{\invariants}{\mathcal{I}}

\newcommand{\paramclock}{x_\mathrm{p}}
\newcommand{\hatparamclock}{\hat x_\mathrm{p}}
\newcommand{\newclock}{z}

\newcommand{\clockval}{\nu}
\newcommand{\paramval}{\gamma}

\newcommand{\N}{\mathbb{N}_0}
\newcommand{\R}{\mathbb{R}_{\geq 0}}
\newcommand{\Rpos}{\mathbb{R}_{> 0}}

\newcommand{\goes}[1]{\ensuremath{\xrightarrow{#1}}}

\newcommand{\spacingHorizontal}{ 3 }
\newcommand{\spacingVertical}{ 3 }

\newcommand{\Regions}{\mathit{Reg}}
\newcommand{\CP}{\mathit{Cp}}

\newcommand{\restr}[2]{#1\mathord{\upharpoonright}_{#2}}

\newcommand{\LESS}{\ensuremath{\mathit{LESS}}}
\newcommand{\MORE}{\ensuremath{\mathit{MORE}}}
\newcommand{\EXACT}{\ensuremath{\mathit{EXACT}}}

\newcommand{\fr}{\ensuremath{\mathit{fr}}}

\newcommand{\suc}{\ensuremath{\mathit{succ}}}

\title{Language Emptiness of Continuous-Time \\ Parametric Timed Automata}

\author{Nikola Bene\v{s}\inst{1}\thanks{%
Nikola Bene\v{s} has been supported by the Czech Science Foundation grant project 
no.~GA15-11089S.
} \and Peter Bezd\v{e}k\inst{1}\thanks{%
Peter Bezd\v{e}k has been supported by the Czech Science Foundation grant project 
no.~GA15-08772S.
} \and Kim G. Larsen\inst{2} \and Ji\v{r}\'{\i} Srba\inst{2}}

\institute{Faculty of Informatics, Masaryk University Brno, Czech Republic 
\and Department of Computer Science, Aalborg University, Denmark}

\begin{document}

\maketitle
\begin{abstract}
Parametric timed automata extend the standard timed automata with the
possibility to use parameters in the clock guards. In general, if the
parameters are real-valued, the problem of language emptiness of such
automata is undecidable even for various restricted subclasses. We thus focus
on the case where parameters are assumed to be integer-valued, while the time
still remains continuous. On the one hand, we show that the problem remains
undecidable for parametric timed automata with three clocks and one parameter.
On the other hand, for the case with arbitrary many clocks where only one of
these clocks is compared with (an arbitrary number of) parameters, we show that
the parametric language emptiness is decidable. The undecidability result
tightens the bounds of a previous result which assumed six parameters, while
the decidability result extends the existing approaches that deal with 
discrete-time semantics only.  To the best of our knowledge, this is the first
positive result in the case of continuous-time and unbounded integer
parameters, except for the rather simple case of single-clock automata.
\end{abstract}

\section{Introduction} \label{sec:intro}
Timed automata~\cite{ad:tcs94} are a popular formalism used for modelling 
of real-time systems. In the classical definition,
the clocks in guards are compared to fixed constants and one of the
key problems, decidable in PSPACE~\cite{ACD:LICS:90}, 
is the question of language 
emptiness. More than 20 years ago, Alur, Henzinger and Vardi~\cite{AHV:STOC:93} 
introduced a parametric variant of the language emptiness problem
where clocks in timed automata
can be additionally compared to a number of parameters. A clock
is \emph{nonparametric} if it is never compared with any 
of the parameters, otherwise the clock is \emph{parametric}.
The parametric language emptiness problem asks whether the parameters
in the system can be replaced by constants so
that the language of the resulting timed automaton becomes nonempty.

Unfortunately, the parametric language emptiness problem is 
undecidable for timed automata with three parametric 
clocks~\cite{AHV:STOC:93}. Yet Alur, Henzinger and Vardi
established a positive decidability result in the case of a single parametric
clock. This decidability result was recently extended by Bundala and 
Ouaknine~\cite{BO:MFCS:14} to the case with two parametric clocks and
an arbitrary number of nonparametric clocks. Both positive results 
are restricted to the discrete-time
semantics with only integer delays. The problem of 
decidability of integer parametric language emptiness in the continuous-time
semantics has been open for over 20 years.
The parametric language emptiness problem has two variants, which
we call \emph{reachability} (existence of a~parameter valuation
s.t.~the language is nonempty) and \emph{safety} (existence of a~parameter
valuation s.t.~the language is empty).
 
Our main contributions, summarised in Table~\ref{tbl:results},
are: (i) undecidability of the reachability and safety
problems (in discrete and continuous-time semantics) for three parametric
clocks, no additional nonparametric clocks and one integer parameter 
and (ii) decidability of the reachability and safety
problems in the continuous-time semantics for one parametric clock with an
arbitrary number of integer parameters and an unlimited number of additional
nonparametric clocks. For reachability the problem is further decidable
in NEXPTIME. 

\begin{table}[t]
\caption{Decidability of the language (non)emptiness problems}\label{tbl:results}
\vspace{3mm}
\centering
\setlength{\tabcolsep}{4pt}
\renewcommand{\arraystretch}{1.1}
\begin{tabular}{l|c|c|c}
&discrete time & continuous time & continuous time \\
&integer parameters & integer parameters & real parameters \\
\hline
$n$ clocks, $m$ parameters & \multirow{2}{*}{decidable~\cite{AHV:STOC:93}} & \multirow{2}{*}{\textbf{decidable}} & \multirow{2}{*}{undecidable~\cite{miller2000decidability}}
\\
1 parametric clock only & & &
\\\hline
3 clocks, 1 parameter & \textbf{undecidable} & \textbf{undecidable} & undecidable~\cite{miller2000decidability}
\\\hline
3 clocks, 6 parameters & undecidable~\cite{AHV:STOC:93}  & undecidable~\cite{AHV:STOC:93}  & undecidable~\cite{AHV:STOC:93} 
\end{tabular}
\vspace{-2mm}
\end{table}

\paragraph{Related work.}
Our undecidability result holds both for discrete and continuous time
semantics and it uses only a single parameter with three parametric clocks,
hence strengthening the result from~\cite{AHV:STOC:93} where six
parameters were necessary for the reduction. 
In~\cite{BO:MFCS:14} the authors established NEXPTIME-completeness
of the parametric reachability problem for the case of a single parametric clock
but only for the discrete-time semantics.
Parametric TCTL model checking of timed automata, in the discrete-time
setting, was also studied in~\cite{BR:FSTTCS:03,Wang1996131}.
Our decision procedure 
for one parametric clock is, to the best of our knowledge, the first one
that deals with continuous-time semantics without any restriction on
the usage of parameters and without bounding the range of the parameters. 

Reachability for parametric timed automata was shown decidable for certain 
(strict) subclasses of parametric timed automata, 
either by bounding the range of  
parameters~\cite{j-TACAS-13} or by imposing syntactic restrictions
on the use of parameters as in L/U automata~\cite{BT:FMSD:09,hune2001linear}.
The study of 
parametric timed automata in continuous time with parameters 
ranging over the rational or real numbers showed undecidability already for one
parametric clock~\cite{miller2000decidability}, or for two parametric
clocks with exclusively strict guards~\cite{Doyen2007208}.
We thus focus solely on integer-valued parameters in this paper.

Parametric reachability problems for interrupt timed automata
were investigated by B\'{e}rard, Haddad, Jovanovi\'{c} and 
Lime~\cite{BHJL:RP:13} with a number of positive decidability results
although their model is incomparable
with the formalism of timed automata studied in this paper. 
Other approaches include the inverse method of~\cite{ACFE:ENTCS:08} 
where the authors
describe a procedure for deriving constrains on  parameters
in order to satisfy that timed automata remain time-abstract
equivalent, however, the termination of the procedure is in general
not guaranteed.

\section{Definitions} \label{sec:definitions}
We shall now introduce parametric timed automata, the studied problems
and give an example of a parametric system for alarm sensor coordination.

Let $\N$ denote the set of nonnegative integers and
$\R$ the set of nonnegative real numbers.
Let $\clocks$ be a finite set of \emph{clocks} and let $\parameters$ be
a finite set of \emph{parameters}. A \emph{simple clock constraint}  
is an expression of the form $x \bowtie c$ where $x \in \clocks$,
$c \in \N \cup \parameters$ and $\bowtie \ \in \{ <, \leq, =, \geq, > \}$. 
A \emph{guard} 
is a conjunction of simple clock constraints, we denote the set of all
guards by $\guards$. A conjunction of simple clock constraints that
contain only upper bounds on clocks, i.e. $\bowtie \ \in \{ <, \leq \}$, 
is called an \emph{invariant} and 
the set of all invariants is denoted by $\invariants$.

A \emph{clock valuation} is a function $\clockval: \clocks \rightarrow \R$
that assigns to each clock its nonnegative real-time age and
\emph{parameter valuation} is a function 
$\paramval: \parameters \rightarrow \N$ that assigns to each parameter
its nonnegative integer value.
Given a clock valuation~$\clockval$, a parameter valuation
$\paramval$ and a guard (or invariant) $g \in \guards$, we write
$\clockval,\paramval \models g$ if the guard expression $g$,
after the substitution of all clocks $x \in \clocks$ with 
$\clockval(x)$ and all parameters $p \in \parameters$ with $\paramval(p)$,
is true.
By $\clockval_0$ we denote the initial clock valuation
where $\clockval_0(x) = 0$ for all $x \in \clocks$. For a clock valuation
$\clockval$ and a delay $d \in \R$, we define the clock valuation
$\clockval+d$ by $(\clockval+d)(x)=\clockval(x)+d$ for all $x \in \clocks$.

\begin{definition}[Parametric Timed Automaton]
A \emph{parametric timed automaton (PTA)} over the set of clocks $\clocks$
and parameters $\parameters$ is a tuple
$A = (\Sigma, L, \ell_0, F, I, \goes{})$ where
$\Sigma$ is a finite \emph{input alphabet}, $L$ is a finite set
of \emph{locations}, $\ell_0 \in L$ is the
\emph{initial location}, $F \subseteq L$ is the set of 
\emph{final (accepting) locations}, $I: L \rightarrow \invariants$
is an \emph{invariant function} assigning invariants to locations, and
$\goes{}\, \subseteq L \times \guards \times \Sigma \times 2^{\clocks} 
\times L$ is the set of transitions, written as
$\ell \goes{g,a,R} \ell'$ whenever $(\ell,g,a,R,\ell') \in \, \goes{}$.
\end{definition}

For the rest of this section, 
let $A = (\Sigma, L, \ell_0, F, I, \goes{})$ be a fixed PTA.
We say that a clock $x \in \clocks$ is a \emph{parametric clock} in $A$
if there is a simple clock constraint of the form $x \bowtie p$ with
$p \in \parameters$ that appears in a~guard or an invariant of~$A$.
Otherwise, if the clock $x$ is never compared to any parameter,
we call it a~\emph{nonparametric clock}.

A configuration of $A$
is a pair $(\ell,\clockval)$ where $\ell \in L$ is the current location
and $\clockval$ is the current clock valuation. For every parameter
valuation $\paramval$ we define the corresponding timed transition
system $T_\paramval(A)$ where states are all configurations $(\ell,\clockval)$
of $A$ that satisfy the location invariants, i.e. 
$\clockval,\paramval \models I(\ell)$, and the transition relation
is defined as follows:
\begin{itemize}
\item $(\ell,\clockval) \goes{d} (\ell,\clockval+d)$ 
 where $d \in \R$ if 
$\clockval+d,\paramval \models I(\ell)$; 
\item $(\ell,\clockval) \goes{a} (\ell',\clockval')$ 
   where $a \in \Sigma$
   if there is a transition $\ell \goes{g,a,R} \ell'$ in $A$
   such that $\clockval,\paramval \models g$ and
   $\clockval',\paramval \models I(\ell')$ where
   for all $x \in \clocks$ we define 
   $\clockval'(x) = 0$ if $x \in R$ and $\clockval'(x) = \clockval(x)$
   otherwise.
\end{itemize}

A \emph{timed language} of $A$ under a parameter 
valuation $\paramval$, denoted by 
$L_\paramval(A)$, is the collection of all accepted \emph{timed words} 
of the form
$(a_0,d_0)(a_1,d_1)\ldots(a_n,d_n) \in (\Sigma\times\R)^*$ 
such that in the transition system $T_\paramval(A)$ there is a computation 
$ (\ell_0,\clockval_0) \goes{d_0} (\ell'_0,\clockval'_0) \goes{a_0}
   (\ell_1,\clockval_1) \goes{d_1} 
\cdots
 \goes{a_{n-1}}
   (\ell_n,\clockval_n) \goes{d_n} (\ell'_n,\clockval'_n) \goes{a_n}
   (\ell_{n+1},\clockval_{n+1})$
where 
$\ell_{n+1} \in F$.

We can now define two problems for parametric timed
automata, namely the reachability problem (reaching desirable locations)
and safety problem (avoiding undesirable locations).
Note that the problems are not completely dual, as the safety problem
contains a hidden alternation of quantifiers.

\begin{problem}[Reachability Problem for PTA]
Given a PTA $A$, is there a parameter valuation $\paramval$ such that
$L_\paramval(A) \not= \emptyset$ ?
\end{problem}

\begin{problem}[Safety Problem for PTA]
Given a PTA $A$, is there a parameter valuation $\paramval$ such that
$L_\paramval(A) = \emptyset$ ?
\end{problem}

We shall now present a small case study of a wireless fire alarm
system~\cite{alarm:2014} modelled as a parametric timed automaton.
In the alarm setup, a number of wireless sensors communicate with the
alarm controller over a limited number of communication channels
(in our simplified example we assume just a single channel).
The wireless alarm system
uses a variant of Time Division Multiple Access (TDMA) protocol in order
to guarantee a safe communication of multiple sensors over a shared
communication channel. In TDMA the data stream is divided into frames
and each frame consists of a number of time slots allocated for exclusive use
by the present wireless sensors. Each sensor is assigned a single slot in
each frame where it can transmit on the shared channel.

\begin{figure}[t]
\centering
\begin{minipage}{0.35\textwidth}
\subfloat[Sensor 1]{\label{fig:sensor1}
        \begin{tikzpicture}[->,>=stealth',auto,node distance=3cm, semithick,scale=0.95,every node/.style={transform shape}]
          \node[accepting,state,align=center,minimum size=1.5cm] (l1) {};
          \node[state,align=center,minimum size=1.5cm] (l2) [right of=l1] {$x_1 < 3$};

          \path[->] (l1) edge [] node {\begin{minipage}{0.55in}\begin{center}$\mathit{wakeup}_1?$ \\$x_1 := 0$\end{center}\end{minipage}} (l2);
          \path[->] (l2) edge [right,bend left=62,distance=21] node [above] {\begin{minipage}{0.65in}\begin{center}$ 2 < x_1 < 3$ \\$\mathit{result_1}!$\end{center}\end{minipage}} (l1);
          \path[->] (l2) edge [loop below,distance=30] node {$\mathit{wakeup}_1?$} (l2);
        \end{tikzpicture}
}

\subfloat[Sensor 2]{\label{fig:sensor2}
        \begin{tikzpicture}[->,>=stealth',auto,node distance=3cm, semithick,scale=0.95,every node/.style={transform shape}]
          \node[accepting,state,align=center,minimum size=1.5cm] (l1) {};
          \node[state,align=center,minimum size=1.5cm] (l2) [right of=l1] {$x_2 < 17$};

          \path[->] (l1) edge [] node {\begin{minipage}{0.65in}\begin{center}$\mathit{wakeup}_2?$\\ $x_2 := 0$\end{center}\end{minipage}} (l2);
          \path[->] (l2) edge [right,bend left=62,distance=21] node [above] {\begin{minipage}{0.75in}\begin{center}$ 2 < x_2 < 3$\\  $\mathit{result}_2!$\end{center}\end{minipage}} (l1);
          \path[->] (l2) edge [right,bend right=62,distance=20] node [above] {\begin{minipage}{0.75in}\begin{center}$ 16 < x_2 < 17$\\ $\mathit{result}_2!$\end{center}\end{minipage}} (l1);
          \path[->] (l2) edge [loop below,distance=30] node {$\mathit{wakeup}_2?$} (l2);
        \end{tikzpicture}
}
\end{minipage}
\begin{minipage}{0.62\textwidth}
\subfloat[Controller with parameters $p_1$ and $p_2$]{\label{fig:controller}
\begin{tikzpicture}[->,>=stealth',auto,node distance=2.8cm, semithick, on grid, scale=0.95,every node/.style={transform shape}]
\clip (-0.89,0.85) rectangle (6.89, -7.01);%
		\node[state,align=center,minimum size=1.5cm] (l3) {$x\leq p_2$\\$y \leq 20$};
    	\node[state,align=center,minimum size=1.5cm] (l4) [below=1.4*\spacingVertical of l3] {$x<2$\\$y \leq 20$};

		\node[accepting,state,align=center,minimum size=1.3cm] (l1) [right=2*\spacingHorizontal of l3] {$x<2$\\$y \leq 20$};
		\node[state,align=center,minimum size=1.5cm] (l2) [below=1.4*\spacingHorizontal of l1] {$x \leq p_1$\\$y \leq 20$};
		
     	\node[state,align=center,minimum size=1.5cm] (l5) [below right=0.7*\spacingVertical and \spacingHorizontal of l3] {$\mathit{fail}$};
      	\node[state,align=center,minimum size=1.3cm] (l6) [below left=0.7*\spacingVertical and \spacingHorizontal of l2] {$\mathit{timeout}$};

		\path[->] (l1) edge [] node [left] {\begin{minipage}{0.49in}$x<2$\\ $\mathit{wakeup}_1!$ \end{minipage}} (l2);
		\path[->] (l2) edge [bend left=5] node [below] {$x<p_1$ $\mathit{result}_1?$ $x:=0$} (l4);
		\path[->] (l2) edge [bend right=5] node [above] {$x=p_1 $ $x:=0$} (l4);
		\path[->] (l4) edge [] node [right]{\begin{minipage}{1.55in}$x<2$\\ $\mathit{wakeup}_2!$\end{minipage}} (l3);
		\path[->] (l3) edge [bend right=5] node [below] {$x<p_2 $ $x:=0$ $y:=0$} (l1);
		\path[->] (l3) edge [bend left=5] node [above] {$x<p_2 $ $ \mathit{result}_2? $ $x:=0$ $y:=0$} (l1);

		\path[->] (l1) edge [bend right=17] node [sloped,below] {$\mathit{result}_1?$} (l5);
		\path[->] (l1) edge [bend left=17] node [sloped,above] {$\mathit{result}_2?$} (l5);
		\path[->] (l2) edge [] node [sloped,above] {$\mathit{result}_2?$} (l5);
		\path[->] (l3) edge [] node [sloped,above] {$\mathit{result}_1?$} (l5);
		\path[->] (l4) edge [bend right=17] node [sloped,above] {$\mathit{result}_2?$} (l5);
		\path[->] (l4) edge [bend left=17] node [sloped,below] {$\mathit{result}_1?$} (l5);

		\path[->] (l1) edge [out=-70,in=0,distance=98] node [above,very near end] {$y=20$} (l6); 
		\path[->] (l2) edge [] node [left] {$y=20$} (l6);
		\path[->] (l3) edge [out=-110,in=180,distance=98] node [above,very near end] {$y=20$} (l6); 
		\path[->] (l4) edge [] node [right] {$y=20$} (l6);
\end{tikzpicture}%
}

\end{minipage}%
 
\caption{Wireless Fire Alarm System}
\label{fig:WFAS}
\vspace{-3mm}
\end{figure}
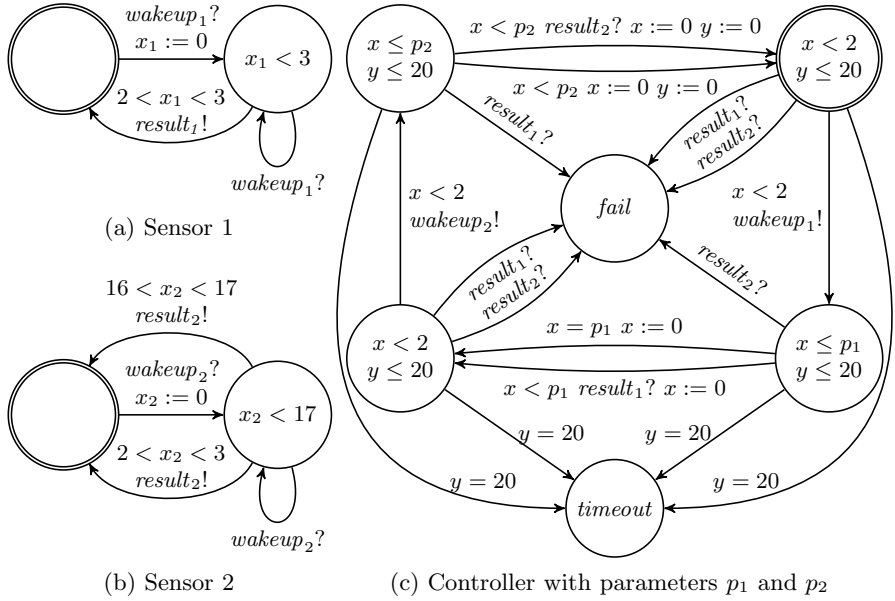

We model each sensor as a timed automaton with two locations as
shown in Figure~\ref{fig:sensor1} and~\ref{fig:sensor2}.
The sensor in Figure~\ref{fig:sensor1} waits in its initial
location until it receives a $\mathit{wakeup}_1$ message from the controller.
After this, it takes strictly between 2 to 3 seconds to gather the
current status of the sensor and transmit it as $\mathit{result}_1$ message back
to the controller. Any subsequent wakeup signals during the transmission 
phase are ignored and after the transmission phase is finished,
the sensor is ready to receive another wakeup signal.
The sensor in Figure~\ref{fig:sensor2} has a more complex behaviour
as transmitting the answer $\mathit{result}_2$ can take either 
strictly between 2 to 3 seconds, or 16 to 17 seconds.

The controller presented in Figure~\ref{fig:controller} is responsible
for synchronising the two sensors and for assigning them their time 
slots so that no transmissions interfere. The parametric clock $x$ of the 
controller determines the size of the time slots. First, it takes at most 
2 seconds for the controller to wake up the first sensor after which it waits
until the elapsed time reaches the value of the parameter~$p_1$.
If it receives the result of the reading of the first
sensor in this time slot, it moves immediately into the next location where it
performs the wakeup of the second sensor. If the first sensor does not
deliver any result and the clock $x$ reaches the value $p_1$, it
also moves to the next location. Now a symmetric control is performed for the
second sensor. If any of the two sensors transmit during the time the
controller transmits the wakeup signals, we enter the
location $\mathit{fail}$. The fail location is also reached if 
$\mathit{result}_2$ is received in the time slot of the first sensor and
vice versa. The second clock $y$ is used to simply measure the duration
of the whole frame; whenever the duration of the frame reaches 20 seconds,
the controller enters the $\mathit{timeout}$ location. 

We assume a standard handshake synchronisation of the controller and
the two sensors running in parallel
that results in a flat product timed automaton with two 
parameters $p_1$ and $p_2$. Note that $x$ is the only parametric clock 
in our example. Now, our task is to find suitable values
of the parameters that guide the duration of the time slots for the
two sensors so that there is no behaviour of the protocol where
it fails or timeouts. This question is equivalent to the safety problem
on the constructed PTA where 
we mark $\mathit{fail}$ and $\mathit{timeout}$ as the accepting 
(undesirable) locations.

The obvious parameter valuation where $\paramval(p_1)=5$ and
$\paramval(p_2)=19$ guarantees that the location $\mathit{fail}$
is unreachable but it is not an acceptable solution as the duration of the
frame becomes $24$ and we reach $\mathit{timeout}$. However, there is another
parameter valuation where $\paramval(p_1)=5$ and $\paramval(p_2)=9$
that guarantees that there is no possibility to fail or timeout. 
This is due to the fact that if the response time of the second sensor
is too long, it skips one slot and the answer fits into 
an appropriate slot in the next frame.

In Section~\ref{sec:decidability} we provide an algorithmic solution for
finding such a parameter valuation that guarantees a~given safety/reachability
criterion.
Note that as we are concerned with language (non)emptiness only, we employ two
simplifications in the rest of the paper: First, we assume that the considered PTA
have no invariants, as moving all invariants to guards preserves the
language. Second, we assume that the alphabet is a~singleton set 
as renaming all actions into a~single action preserves language (non)emptiness.

\section{Undecidability for Three Parametric Clocks} \label{sec:undecidability}
We shall now provide a reduction from the halting/boundedness problems of
two counter Minsky machine to the reachability/safety problems on 
PTA.
A~\emph{Minsky machine} with two nonnegative counters $c_1$ and $c_2$ is
a sequence of labelled instructions
$1:\mathit{inst}_1;\ 2:\mathit{inst}_2;\ \ldots, n:\mathit{inst}_n$
where $\mathit{inst}_n = \mathit{HALT}$ 
and each $\mathit{inst}_i$, $1\le i < n$, is of one of the following forms
(for $r\in \{1,2\}$ and $1\le j,k \le n$):
\begin{itemize}
\newlength{\tablen}
\settowidth\tablen{(Test and Decrement) \quad}
\item
\makebox[\tablen][l]{(Increment)}
\texttt{$i$: $c_r$++; goto $j$}
\item
\makebox[\tablen][l]{(Test and Decrement)}
\texttt{$i$: if $c_r$=0 then goto $k$ else ($c_r$--; goto $j$)}
\end{itemize}

A configuration is a triple $(i, v_1, v_2)$ where $i$ is the current
instruction and $v_1, v_2 \in \N$ are the values of the counters $c_1$
and $c_2$, respectively. A computation step between configurations is
defined in the natural way. If starting from the initial
configuration $(1,0,0)$ the machine
reaches the instruction $\mathit{HALT}$ (note that the computation
is deterministic)
then we say it \emph{halts}, otherwise it \emph{loops}.
The problem whether a given Minsky machine halts
is undecidable~\cite{Minsky:book}. The boundedness problem,
i.e. the question whether there is a constant $K$ such that
$v_1+v_2 \leq K$ for any configuration $(i,v_1,v_2)$ reachable
from $(1, 0, 0)$, is also undecidable~\cite{KCH:ACCS:10}.

The reduction from a two counter Minsky machine to PTA with a single
parameter $p$ and three parametric clocks $x_1$, $x_2$ and $z$ is
depicted in Figure~\ref{fig:encoding}. The reduction rules are shown
only for the instructions handling the first
counter. The rules for the
second counter are symmetric.
We also omit the transition labels as they are not relevant for the emptiness
problem. The reduction preserves the property that whenever
we are in a configuration $(\ell_i,\clockval)$ where
$\clockval(z)=0$ then $\clockval(x_1)$ and $\clockval(x_2)$ represent
the exact values of the counters $c_1$ resp. $c_2$, and
the next instruction to be executed is the one with label $i$.
Note also that there are no invariants used in the constructed automaton.

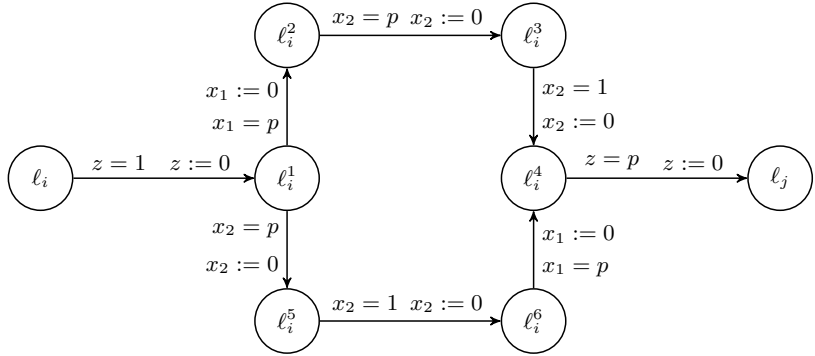
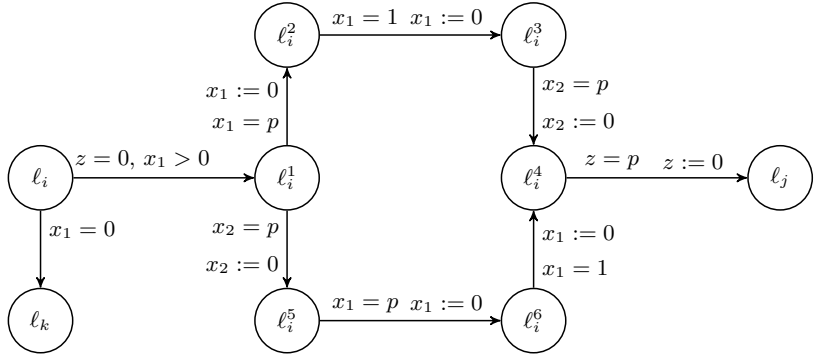
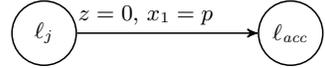
\begin{figure}[t!]
\newcommand{\FigScale}{0.95}
\centering
\subfloat[Increment \ \ {\tt $i$: $c_1$++; goto $j$}]{\label{fig:encoding:inc}
\scalebox{\FigScale}{
\begin{tikzpicture}[->,>=stealth',auto,node distance=5.5cm, semithick, xscale=3.45, yscale=2.0]
  \node[state,align=center,minimum size=.9cm] (li) at (0,1) {$\ell_i$};
  \node[state,align=center,minimum size=.9cm] (li1) at (1,1) {$\ell^1_i$};
  \node[state,align=center,minimum size=.9cm] (li2) at (1,2) {$\ell^2_i$};
  \node[state,align=center,minimum size=.9cm] (li3) at (2,2) {$\ell^3_i$};
  \node[state,align=center,minimum size=.9cm] (li4) at (2,1) {$\ell^4_i$};
  \node[state,align=center,minimum size=.9cm] (li5) at (1,0) {$\ell^5_i$};
  \node[state,align=center,minimum size=.9cm] (li6) at (2,0) {$\ell^6_i$};
  \node[state,align=center,minimum size=.9cm] (lj) at (3,1) {$\ell_j$};
  \path[->] (li) edge [] node[pos=0.25] {$z=1$}
                         node[pos=0.7] {$z:=0$} (li1);
  \path[->] (li1) edge [] node[left,pos=0.25] {$x_1=p$}
                         node[left,pos=0.7] {$x_1:=0$} (li2);
  \path[->] (li2) edge [] node[pos=0.25] {$x_2=p$}
                         node[pos=0.7] {$x_2:=0$} (li3);
  \path[->] (li3) edge [] node[right,pos=0.25] {$x_2=1$}
                         node[right,pos=0.7] {$x_2:=0$} (li4);
  \path[->] (li4) edge [] node[pos=0.25] {$z=p$}
                         node[pos=0.7] {$z:=0$} (lj);
  \path[->] (li1) edge [] node[left,pos=0.25] {$x_2=p$}
                         node[left,pos=0.7] {$x_2:=0$} (li5);
  \path[->] (li5) edge [] node[pos=0.25] {$x_2=1$}
                         node[pos=0.7] {$x_2:=0$} (li6);
  \path[->] (li6) edge [] node[right,pos=0.25] {$x_1=p$}
                         node[right,pos=0.7] {$x_1:=0$} (li4);
\end{tikzpicture}
}
}

\subfloat[Test and decrement \ \ 
{\tt $i$: if $c_1$=$0$ then goto $k$ else ($c_1$--; goto $j$)}
]{\label{fig:encoding:dec}
\scalebox{\FigScale}{
\begin{tikzpicture}[->,>=stealth',auto,node distance=5.5cm, semithick, xscale=3.45, yscale=2.0]
  \node[state,align=center,minimum size=.9cm] (li) at (0,1) {$\ell_i$};
  \node[state,align=center,minimum size=.9cm] (li1) at (1,1) {$\ell^1_i$};
  \node[state,align=center,minimum size=.9cm] (li2) at (1,2) {$\ell^2_i$};
  \node[state,align=center,minimum size=.9cm] (li3) at (2,2) {$\ell^3_i$};
  \node[state,align=center,minimum size=.9cm] (li4) at (2,1) {$\ell^4_i$};
  \node[state,align=center,minimum size=.9cm] (li5) at (1,0) {$\ell^5_i$};
  \node[state,align=center,minimum size=.9cm] (li6) at (2,0) {$\ell^6_i$};
  \node[state,align=center,minimum size=.9cm] (lj) at (3,1) {$\ell_j$};
  \node[state,align=center,minimum size=.9cm] (lk) at (0,0) {$\ell_k$};
  \path[->] (li) edge [] node[pos=0.25] {$x_1=0$}
                         node[pos=0.7] {} (lk);
  \path[->] (li) edge [] node[pos=0.38] {$z=0$, $x_1>0$}
                         node[pos=0.7] {} (li1);
  \path[->] (li1) edge [] node[left,pos=0.25] {$x_1=p$}
                         node[left,pos=0.7] {$x_1:=0$} (li2);
  \path[->] (li2) edge [] node[pos=0.25] {$x_1=1$}
                         node[pos=0.7] {$x_1:=0$} (li3);
  \path[->] (li3) edge [] node[right,pos=0.25] {$x_2=p$}
                         node[right,pos=0.7] {$x_2:=0$} (li4);
  \path[->] (li4) edge [] node[pos=0.25] {$z=p$}
                         node[pos=0.7] {$z:=0$} (lj);
  \path[->] (li1) edge [] node[left,pos=0.25] {$x_2=p$}
                         node[left,pos=0.7] {$x_2:=0$} (li5);
  \path[->] (li5) edge [] node[pos=0.25] {$x_1=p$}
                         node[pos=0.7] {$x_1:=0$} (li6);
  \path[->] (li6) edge [] node[right,pos=0.25] {$x_1=1$}
                         node[right,pos=0.7] {$x_1:=0$} (li4);
\end{tikzpicture}
}
}

\subfloat[For safety, add this for every instruction {\tt $i$: $c_1$++; goto $j$}]{
\label{fig:encoding:incsafety}
\scalebox{\FigScale}{
\begin{tikzpicture}[->,>=stealth',auto,node distance=3.45cm, semithick]
  \node[state,align=center,minimum size=.9cm] (li) {$\ell_{j}$};
  \node[] (left) [left of=li] {};
  \node[state,align=center,minimum size=.9cm] (lacc) [right of=li] {$\ell_{\mathit{acc}}$};
\node[] (right) [right of=lacc] {};
  \path[->] (li) edge [] node[pos=0.38] {$z=0$, $x_1=p$} (lacc);
\end{tikzpicture}
}
}
\caption{Encoding of Minsky Machine as PTA with a single parameter $p$}
\label{fig:encoding}
\vspace{-3mm}
\end{figure}

\begin{lemma}\label{lem:minsky}
Let $M$ be a Minsky machine.
Let $A$ be the PTA built according to the rules in Figures~\ref{fig:encoding:inc} and \ref{fig:encoding:dec} (without the 
transitions for safety) and where $\ell_1$ is the initial 
location and $\ell_n$ is the only accepting location.
The Minsky machine $M$ halts iff there is a parameter
valuation $\paramval$ such that $L_\paramval(A) \not= \emptyset$.
\end{lemma}
\begin{proof}[Sketch]
We only sketch a~part of the proof to show the basic idea.
We argue that from the configuration $(\ell_i,\clockval)$
where $\clockval(z)=0$ and where $\clockval(x_1)$ 
and $\clockval(x_2)$ represent the counter values, there is a unique 
way to move from $\ell_i$ to $\ell_j$ (or possibly also to $\ell_k$ in the case of 
the test and decrement instruction) where again $\clockval(z)=0$ and the counter
values are updated accordingly. As there are no invariants in the automaton,
we can always delay long enough so that we get stuck in a given location,
but this behaviour will not influence the language emptiness problem we
are interested in.

Consider the automaton for the increment instruction from
Figure~\ref{fig:encoding:inc} and assume we are in a configuration 
$(\ell_i,\clockval)$ where $\clockval(z)=0$, $\clockval(x_1)=v_1$ 
and $\clockval(x_2)=v_2$.
First note that if $v_1\geq p$ then there is no execution ending in $\ell_k$
due to the forced delay of one time unit on the transition from
$\ell_i$ to $\ell_i^1$ and the guard $x_1=p$ tested in both the upper
and lower branch in the automaton. Assume thus that $v_1 < p$.
If $v_1 \geq v_2$ then we can perform the following execution
with uniquely determined time delays:
$(\ell_i, [x_1 \mapsto v_1, x_2 \mapsto v_2, z \mapsto 0])
\goes{1}
 (\ell_i^1, [x_1 \mapsto v_1+1, x_2 \mapsto v_2+1, z \mapsto 0])
\goes{p-v_1-1}
 (\ell_i^2, [x_1 \mapsto 0, x_2 \mapsto p - v_1 + v_2, z \mapsto p-v_1-1])
\goes{v_1-v_2}
 (\ell_i^3, [x_1 \mapsto v_1-v_2, x_2 \mapsto 0, z \mapsto p-v_2-1])
\goes{1}
 (\ell_i^4, [x_1 \mapsto v_1-v_2+1, x_2 \mapsto 0, z \mapsto p-v_2])
\goes{v_2}
 (\ell_j, [x_1 \mapsto v_1+1, x_2 \mapsto v_2, z \mapsto 0])$.
In this case where $v_1 \geq v_2$, executing the lower branch of the automaton
will result in getting stuck in the location $\ell_i^6$ as here
necessarily $\clockval(x_1)>p$.
Clearly, there is a~unique way of getting to $\ell_j$ in which the
clock valuation of $x_1$ was incremented by one, hence faithfully simulating
the increment instruction of the Minsky machine.
The other cases and instructions are dealt with similarly, see Appendix~\ref{app:undec}. 
\qed
\end{proof}

\begin{lemma}
Let $M$ be a Minsky machine.
Let $A$ be the PTA built according to the rules in Figures~\ref{fig:encoding:inc}, \ref{fig:encoding:dec} and \ref{fig:encoding:incsafety} (including the transitions for safety) and where $\ell_1$ is the initial 
location and $\ell_{\mathit{acc}}$ is the only accepting location.
The Minsky machine $M$ is bounded iff there is a parameter
valuation $\paramval$ such that $L_\paramval(A) = \emptyset$.
\end{lemma}
\begin{proof}
If the computation of the Minsky machine is unbounded then clearly,
for any parameter value of $p$, the Minsky machine 
will eventually try to make one of the counters larger or equal than $p$ 
(using the
increment instruction). Necessarily, we will then have $\clockval(x_1) = p$ 
or $\clockval(x_2) = p$ in the location 
$\ell_j$ where we end after performing the increment instruction $i$, 
implying that we can reach the accepting 
location $\ell_{\mathit{acc}}$ due to the transition added
in Figure~\ref{fig:encoding:incsafety} and hence the language is nonempty.
On the other hand, if the parameter $p$ is large enough and the computation
bounded (note that the boundedness condition $\exists K.\ v_1 + v_2 \leq K$
is equivalent to $\exists K.\ \max\{v_1,v_2\} \leq K$), we will not be
able to enter the accepting location $\ell_{\mathit{acc}}$
and the language is empty. \qed
\end{proof}

We now conclude with the main theorem of this section, tightening
the previously known undecidability result that used six parameters and
three parametric clocks~\cite{AHV:STOC:93}.
The theorem is valid for both the continuous-time and 
the discrete-time semantics due to the
exact guards in all transitions of the constructed PTA that allow 
to take transitions only after integer delays.

\begin{theorem}
The reachability and safety problems are undecidable for PTA with one integer
parameter, three parametric clocks and no further nonparametric clocks in the
continuous-time as well as the discrete-time semantics.
\end{theorem}

\section{Decidability for One Parametric Clock} \label{sec:decidability}

In this section, we show that both the reachability and safety problems for PTA
with a~single parametric clock are decidable. Our general strategy is similar
to that of \cite{AHV:STOC:93}, i.e.~reducing the original PTA (which has
continuous-time semantics in our case) into a~so-called
\emph{parametric 0/1-timed automaton} with just a~single clock. It is shown in
\cite{AHV:STOC:93} that the set of parameter valuations that ensure language
nonemptiness of a~given parametric 0/1-timed automaton with single clock is
effectively computable. Moreover, in \cite{BO:MFCS:14} the authors show that
the reachability problem for parametric 0/1-timed automata is polynomial-time
reducible to the
halting problem of parametric bounded one-counter machines, which is in NP. As
the parametric 0/1-timed automaton is going to be exponential in the size of
the original PTA, this makes the reachability problem for PTA with a~single
parametric clock belong to the NEXPTIME complexity class.

A~0/1-timed automaton is a~timed automaton with discrete time, in which all the
delays are explicitly encoded via two kinds of delay transitions: 0-transitions
and 1-transitions. Formally, we enrich the syntax of a~timed automaton with
two transition relations $\goes{0}$, $\mathord{\goes{1}} \subseteq 
L \times L$ and modify the semantics so that $(\ell,\nu) \goes{0} (\ell',\nu)$
iff $\ell \goes{0} \ell'$ and $(\ell,\nu) \goes{1} (\ell',\nu+1)$ iff
$\ell \goes{1} \ell'$; other delays in the timed transition system are
no longer possible.

Note that this treatment of $\goes{0}$ and $\goes{1}$ as special transitions
differs slightly from the original definition of~\cite{AHV:STOC:93}, in which
a~0/1 label is given to every transition of the 0/1-timed automaton. This change
is only cosmetic; the ability to distinguish between 0/1-transitions and action
transitions will be useful in later proofs.

\paragraph{Corner-Point Abstraction.}
As we are concerned with continuous time, our reduction to 0/1-timed
automata is more convoluted than that of~\cite{AHV:STOC:93}, in which
the nonparametric clocks were eliminated by moving their integer values
into locations. In our setting, using region abstraction to eliminate 
nonparametric clocks will not allow us to correctly
identify the 0/1 delays. We thus choose to use 
\emph{corner-point abstraction}~\cite{BFHLPRV:HSCC:01} that 
is finer than the region-based one.
In this abstraction, each region is associated with a~set
of its corner points. Note that the original definition only deals with
timed automata that are bounded, while we want to be more general here.
For this reason, we extend the original definition with
extra corner points for unbounded regions.

We first define the region equivalence~\cite{ad:tcs94}.
Let $M \in \N$ be the largest constant appearing in the constraints of 
a~given timed automaton.
Note that in the original definition the largest constant is considered for
each clock independently. For the sake of readability, we consider $M$ to be
a~common upper bound for each clock.
Let $\clockval,\clockval'$ be clock valuations.
Let further $\fr(t)$ be the fractional part of $t$ and $\lfloor t \rfloor$ be the integral part of $t$.  
We define an equivalence relation $\equiv$ on clock valuations
by $\clockval \equiv \clockval'$  if and only if the following three conditions are satisfied:
\begin{itemize}
\item for all $x \in \clocks$ either $\clockval(x)\geq M$ and $\clockval'(x)\geq M$ or $\lfloor\clockval(x)\rfloor=\lfloor\clockval'(x)\rfloor$;
\item for all $x, y \in \clocks$ such that $\clockval(x)\leq M$ and $\clockval(y)\leq M$, $\fr(\clockval(x)) \leq \fr(\clockval(y))$ if and only if $\fr(\clockval'(x)) \leq \fr(\clockval'(y))$;
\item for all $x \in \clocks$ such that $\clockval(x)\leq M$, $\fr(\clockval(x))=0$ if and only if $\fr(\clockval'(x))=0$. 
\end{itemize}
We define a \textit{region} as an~equivalence class of clock valuations induced by $\equiv$.
A~region $r'$ is a~time successor of a~region $r$ if for all $\clockval
\in r$ there exists $d \in \Rpos$ such that $\clockval+d \in r'$ and for all
$d'$, $0
\leq d' \leq d$, we have $\clockval +d' \in r \cup r'$. As the time successor is
unique if it exists, we use $\suc(r)$ to denote the time successor of $r$.
Moreover, if no time successor of $r$ exists, we let $\suc(r) = r$.

An \textit{$(M^{+1})$-corner point} $\alpha: \clocks \goes{} \N\cap[0,M+1]$ is a
function which assigns an integer value from the interval $[0,M+1]$ to each clock. 
We define the successor of the $M^{+1}$-corner point $\alpha$,
denoted by $\suc(\alpha)$, as follows:
\[
\text{ for each }x \in \clocks,\ \ succ(\alpha)(x)= 
\begin{cases} 
\alpha(x) + 1  &  \alpha(x) \leq M \\
M+1  & \text{otherwise .} 
\end{cases} 
\]
For $R\subseteq \clocks$, we define the reset of the 
corner point $\alpha$, denoted by $\alpha[R]$, as follows:
\[
\text{ for each } x \in \clocks,\ \ \alpha[R](x)= 
\begin{cases} 
\alpha(x)  &  x \not \in R \\
0  & x \in R \ .
\end{cases} 
\]
We say $\alpha$ is a~corner point of a~region $r$ 
whenever $\alpha$ is in the topological closure of~$r$. 
The construction of the corner-point abstraction is illustrated in
Figure~\ref{fig:cpa}. Notice the corner points in unbounded regions.

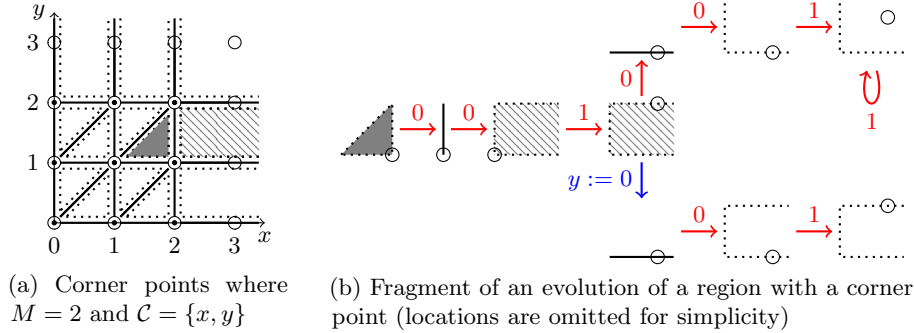
\begin{figure}[t]
\newcommand{\CPRadius}{3}
\centering
\begin{minipage}{0.35\textwidth}
\subfloat[Corner points where $M=2$ and $\clocks=\{x,y\}$]{\label{fig:cornerPoints}
\begin{tikzpicture}[scale=0.8]

\draw [<->] (0,3.5) node (yaxis) [left] {$y$}
        |- (3.5,0) node (xaxis) [below] {$x$};
\foreach \x in {0,1}
	\foreach \y in {0,1}
	    \draw [-, dotted, thick] (\x+0.2,\y+0.1) -- (\x+0.9,\y+0.1) -- (\x+0.9,\y+0.8) -- (\x+0.2,\y+0.1);

\foreach \x in {0,1}
	\foreach \y in {0,1}
	    \draw [-, dotted, thick] (\x+0.1,\y+0.2) -- (\x+0.1,\y+0.9) -- (\x+0.8,\y+0.9) -- (\x+0.1,\y+0.2);

\foreach \x in {0,1,2}
	\foreach \y in {0,1,2}
		\filldraw (\x,\y) circle (1pt);			    

\foreach \x in {0,1,2} 
	\foreach \y in {0,1,2} {
	    \draw [-, solid, thick] (\x+0.1,\y) -- (\x+0.9,\y+0);
	    \draw [-, solid, thick] (\x,\y+0.1) -- (\x,\y+0.9);
	}
\foreach \x in {0,1} 
	\foreach \y in {0,1} {
	    \draw [-, solid, thick] (\x+0.1,\y+0.1) -- (\x+0.9,\y+0.9);
	    \draw [-, solid, thick] (\x+0.1,\y+0.1) -- (\x+0.9,\y+0.9);
	}
\foreach \x in {0,1} 
    \draw [-, dotted, thick] (\x+0.1,3.4) -- (\x+0.1,2.1) -- (\x+0.9,2.1) -- (\x+0.9,3.4);
\foreach \y in {0,1} 
    \draw [-, dotted, thick] (3.4,\y+0.1) -- (2.1,\y+0.1) -- (2.1,\y+0.9) -- (3.4,\y+0.9);
\draw [-, dotted, thick] (2.1,3.4) -- (2.1,2.1) -- (3.4,2.1);			

\foreach \x in {0,1,2} 
    \draw [-, solid, thick] (\x,3.4) -- (\x,2.1);
\foreach \y in {0,1,2} 
    \draw [-, solid, thick] (3.4,\y) -- (2.1,\y);

\foreach \x in {0,1,2,3}
    \draw[-] (\x ,-3pt) -- (\x ,-3pt) node[anchor=north] {$\x$};
\foreach \y in {1,2,3}
    \draw[-] (-3pt,\y ) -- (-3pt,\y ) node[anchor=east] {$\y$};
    
\foreach \x in {0,1,2,3} 
	\foreach \y in {0,1,2,3} {
    \draw (\x,\y) circle (\CPRadius pt);
	}
	
    \draw [-, dotted, thick, fill=gray] (1.2,1.1) -- (1.9,1.1) -- (1.9,1.8);
    \draw [-, dotted, thick, pattern=north west lines, pattern color=gray] (3.4,1.1) -- (2.1,1.1) -- (2.1,1.9) -- (3.4,1.9);
   
\end{tikzpicture}
}
\end{minipage}%
\begin{minipage}{0.65\textwidth}
\subfloat[Fragment of an evolution of a region with a corner point 
(locations are omitted for simplicity)]{ \label{fig:cornerPointsEvolution}
\begin{tikzpicture}[scale=0.85]
\newcommand{\RegSize}{0.8}
\newcommand{\URegSize}{01}
\newcommand{\ArrowSize}{0.8}
\newcommand{\ArrowYOffset}{\RegSize/2}


	\draw [-, dotted, thick, fill=gray] (0,0) -- (\RegSize,\RegSize) -- (\RegSize,0)
	--cycle;
    \draw (\RegSize,0) circle (\CPRadius pt);
    
    \path[->, red, thick] (\RegSize,\ArrowYOffset) edge [shorten <= 0.1cm, shorten >= 0.1cm] node[above=0pt] {0} (\RegSize+\ArrowSize,\ArrowYOffset);
    
	\draw [-, solid, thick] (\RegSize+\ArrowSize,0) -- (\RegSize+\ArrowSize,\RegSize);
    \draw (\RegSize+\ArrowSize,0) circle (\CPRadius pt);

    \path[->, red, thick] (\RegSize+\ArrowSize,\ArrowYOffset) edge [ shorten <= 0.1cm, shorten >= 0.1cm] node[above=0pt] {0} (\RegSize+2*\ArrowSize,\ArrowYOffset);
        	
	\draw [-, dotted, thick, pattern=north west lines, pattern color=gray] (\URegSize+\RegSize+2*\ArrowSize,0) -- (\RegSize+2*\ArrowSize,0) -- (\RegSize+2*\ArrowSize,\RegSize) -- (\URegSize+\RegSize+2*\ArrowSize,\RegSize);
    \draw (\RegSize+2*\ArrowSize,0) circle (\CPRadius pt);

    \path[->, red, thick] (\URegSize+\RegSize+2*\ArrowSize,\ArrowYOffset) edge [ shorten <= 0.1cm, shorten >= 0.1cm] node[above=0pt] {1} (\URegSize+\RegSize+3*\ArrowSize,\ArrowYOffset);
        
	\draw [-, dotted, thick,  pattern=north west lines, pattern color=gray] (2*\URegSize+\RegSize+3*\ArrowSize,0) -- (\URegSize+\RegSize+3*\ArrowSize,0) -- (\URegSize+\RegSize+3*\ArrowSize,\RegSize) -- (2*\URegSize+\RegSize+3*\ArrowSize,\RegSize);
    \draw (2*\URegSize+\RegSize+3*\ArrowSize-0.25,\RegSize) circle (\CPRadius pt);

    \path[->, blue, thick] (1.5*\URegSize+\RegSize+3*\ArrowSize,0) edge [ shorten <= 0.1cm, shorten >= 0.1cm] node[left=0pt] {$y:=0$} (1.5*\URegSize+\RegSize+3*\ArrowSize,-\ArrowSize);
    
    \draw [-, solid, thick] (2*\URegSize+\RegSize+3*\ArrowSize,-\RegSize-\ArrowSize) -- (1*\URegSize+\RegSize+3*\ArrowSize,-\RegSize-\ArrowSize) ;
    \draw (2*\URegSize+\RegSize+3*\ArrowSize-0.25,-\RegSize-\ArrowSize) circle (\CPRadius pt);
    
    \path[->, red, thick] (2*\URegSize+\RegSize+3*\ArrowSize,-0.5*\RegSize-\ArrowSize) edge [ shorten <= 0.1cm, shorten >= 0.1cm] node[above=0pt] {0} (2*\URegSize+\RegSize+4*\ArrowSize,-0.5*\RegSize-\ArrowSize);
        
    \draw [-, dotted, thick] (3*\URegSize+\RegSize+4*\ArrowSize,-\RegSize-\ArrowSize) -- (2*\URegSize+\RegSize+4*\ArrowSize,-\RegSize-\ArrowSize) -- (2*\URegSize+\RegSize+4*\ArrowSize,-\ArrowSize) -- (3*\URegSize+\RegSize+4*\ArrowSize,-\ArrowSize);
    \draw (3*\URegSize+\RegSize+4*\ArrowSize-0.25,-\RegSize-\ArrowSize) circle (\CPRadius pt);

    \path[->, red, thick] (3*\URegSize+\RegSize+4*\ArrowSize,-0.5*\RegSize-\ArrowSize) edge [ shorten <= 0.1cm, shorten >= 0.1cm] node[above=0pt] {1} (3*\URegSize+\RegSize+5*\ArrowSize,-0.5*\RegSize-\ArrowSize);
        
	\draw [-, dotted, thick] (4*\URegSize+\RegSize+5*\ArrowSize,-\ArrowSize) -- (3*\URegSize+\RegSize+5*\ArrowSize,-\ArrowSize) -- (3*\URegSize+\RegSize+5*\ArrowSize,-\RegSize-\ArrowSize) -- (4*\URegSize+\RegSize+5*\ArrowSize,-\RegSize-\ArrowSize);
    \draw (4*\URegSize+\RegSize+5*\ArrowSize-0.25,-\ArrowSize) circle (\CPRadius pt);
    
    \path[->, red, thick] (1.5*\URegSize+\RegSize+3*\ArrowSize,\RegSize) edge [ shorten <= 0.1cm, shorten >= 0.1cm] node[left=0pt] {0} (1.5*\URegSize+\RegSize+3*\ArrowSize,\RegSize+\ArrowSize);
        
	\draw [-, solid, thick] (2*\URegSize+\RegSize+3*\ArrowSize,\RegSize+\ArrowSize) -- (1*\URegSize+\RegSize+3*\ArrowSize,\RegSize+\ArrowSize) ;
    \draw (2*\URegSize+\RegSize+3*\ArrowSize-0.25,\RegSize+\ArrowSize) circle (\CPRadius pt);

    \path[->, red, thick] (2*\URegSize+\RegSize+3*\ArrowSize, 1.5*\RegSize+\ArrowSize) edge [ shorten <= 0.1cm, shorten >= 0.1cm] node[above=0pt] {0} (2*\URegSize+\RegSize+4*\ArrowSize, 1.5*\RegSize+\ArrowSize);
    
	\draw [-, dotted, thick] (3*\URegSize+\RegSize+4*\ArrowSize, \RegSize+\ArrowSize) -- (2*\URegSize+\RegSize+4*\ArrowSize,\RegSize+\ArrowSize) -- (2*\URegSize+\RegSize+4*\ArrowSize,2*\RegSize+\ArrowSize);
    \draw (3*\URegSize+\RegSize+4*\ArrowSize-0.25,\RegSize+\ArrowSize) circle (\CPRadius pt);

    \path[->, red, thick] (3*\URegSize+\RegSize+4*\ArrowSize, 1.5*\RegSize+\ArrowSize) edge [ shorten <= 0.1cm, shorten >= 0.1cm] node[above=0pt] {1} (3*\URegSize+\RegSize+5*\ArrowSize, 1.5*\RegSize+\ArrowSize);

	\draw [-, dotted, thick] (4*\URegSize+\RegSize+5*\ArrowSize, \RegSize+\ArrowSize) -- (3*\URegSize+\RegSize+5*\ArrowSize, \RegSize+\ArrowSize) -- (3*\URegSize+\RegSize+5*\ArrowSize, 2*\RegSize+\ArrowSize);
    \draw (4*\URegSize+\RegSize+5*\ArrowSize-0.25, 2*\RegSize+\ArrowSize-0.25) circle (\CPRadius pt);
        
    \path[->, red, thick] (3.5*\URegSize+\RegSize+5*\ArrowSize, \RegSize+\ArrowSize) edge [ shorten <= 0.25cm, shorten >= 0.25cm, loop below, distance=1cm] node[below=2pt] {1} (3.5*\URegSize+\RegSize+5*\ArrowSize, \RegSize+\ArrowSize);
        

\end{tikzpicture}
}
\end{minipage}
\caption{Corner point abstraction}\label{fig:cpa}
\end{figure}

\paragraph{Construction of the Parametric 0/1-Timed Automaton.}
Now we show how to construct for a given PTA with one parametric clock an
equivalent 0/1-PTA with just one clock.
Let $A=(\Sigma,L,\ell_0,F,I,\goes{})$ be the original PTA over the set of clocks 
$\clocks$ and parameters $\parameters$. Let $\paramclock$ denote the only parametric
clock.

We first modify the automaton by adding a~fresh clock $\newclock$ as follows:
every transition $\ell \goes{g,a,R} \ell'$ is changed into 
$\ell \goes{g \wedge \newclock < 1,a,R'} \ell'$ where $R' = R$ if $\paramclock \not\in R$, and
$R' = R \cup \{\newclock\}$ otherwise.
To every location $\ell$ we then add a~new self-loop transition
$\ell \goes{\newclock = 1,a,\{\newclock\}} \ell$.
Intuitively, the new clock $\newclock$ will always contain 
the fractional part of $\paramclock$.
We call this new automaton $A'$.
Clearly, this modification preserves the language (non)emptiness of the original 
automaton~$A$.

In the second step, we use the corner-point abstraction of $A'$ with respect to
all clocks except for $\paramclock$ to create the 0/1-timed automaton with a~single clock.
Let $\hat\clocks = (\clocks \cup \{\newclock\}) \setminus \{\paramclock\}$ and
let $M$ be the largest constant appearing in the guards concerning the clocks 
in~$\hat\clocks$. In the following, we consider regions and corner-points with respect
to clocks in $\hat\clocks$ and the bound $M$. Let $\Regions$ denote the set of
all such regions and let $\CP$ denote the set of all corresponding corner-points,
i.e.~$\CP = (\N \cap [0,M+1])^{\hat\clocks}$.

We use the following auxiliary notation. Let $r \in \Regions$
and $\alpha \in \CP$.
\[ \iota(r,\alpha) = \begin{cases}
	\LESS & \alpha(\newclock) = 1 \text{ and } r \not\models \newclock = 1\\
	\MORE & \alpha(\newclock) = 0 \text{ and } r \not\models \newclock = 0\\
	\EXACT&\text{otherwise}
\end{cases}\]

The 0/1-timed automaton over the singleton set of clocks $\{\hatparamclock\}$
is $\hat A = (\Sigma, L \times \Regions \times \CP,
(\ell_0,r_0,\alpha_0), F \times \Regions \times \CP, I, \goes{})$
where $r_0$ is the initial region and $\alpha_0(x) = 0$ for all $x \in \hat\clocks$ is the initial corner-point.
The transition relation is defined as follows:
\begin{itemize}
\item zero delay: $(\ell,r,\alpha) \goes{0} (\ell,r',\alpha)$ if $r' = succ(r)$ and 
$\alpha$ is a~corner-point of both $r$ and $r'$;
\item unit delay: $(\ell,r,\alpha) \goes{1} (\ell,r,\alpha')$ if $\alpha' = succ(\alpha)$
and both $\alpha$ and $\alpha'$ are corner-points of $r$;
\item action: whenever $\ell \goes{g,a,R} \ell'$ in $A'$ then
let $g_1$, \ldots, $g_k$ be all the simple clock constraints appearing in $g$ comparing
clocks from $\hat\clocks$ and let $h_1$, \ldots, $h_n$ be the remaining simple clock
constraints, i.e.~those that consider $\paramclock$.
For every $(\ell,r,\alpha)$ that satisfies 
(1) $r \models g_1 \wedge \cdots \wedge g_k$
and
(2) if $\iota(r,\alpha) \ne \EXACT$ then no $h_i$ contains equality (=),
we set $(\ell,r,\alpha) \goes{\hat{h}_1\wedge\cdots\wedge\hat{h}_n, \hat R}
(\ell',r[R \setminus \{\paramclock\}], \alpha[R\setminus\{\paramclock\}])$,
where $\hat R = \{\hatparamclock\}$ if $\paramclock \in R$ and $\hat R = \emptyset$ otherwise.
The constraints
$\hat h_i$ are created as follows:
all $\paramclock$ are changed into $\hatparamclock$;
if $\iota(r,\alpha) = \LESS$, all $<$ are changed into $\le$ and all $\ge$ are changed
into $>$;
if $\iota(r,\alpha) = \MORE$, all $\le$ are changed into $<$ and all $>$ are changed
into $\ge$.
\end{itemize}



\begin{theorem} \label{thm:main}
The reachability and safety problems for parametric timed automata over integer
parameters with one parametric clock in the continuous-time semantics are
decidable. Moreover, the reachability problem is in NEXPTIME.
\end{theorem}
\begin{proof}[Idea]
Due to space constraints, the complete proof can be found in Appendix~\ref{app:dec}.
As mentioned above, the modification from $A$ to $A'$ 
preserves the language (non)emptiness.
The idea of the proof is to show that for every given parameter valuation,
every run of $A'$ has a~corresponding run in $\hat A$ and vice versa.
This shows that the reachability and safety problems for parametric timed
automata with one parametric clock reduce to the reachability and safety problems
for parametric 0/1-timed automata. These problems were shown decidable
in~\cite{AHV:STOC:93}. The complexity argument is discussed
in the beginning of this section.\qed
\end{proof}

\section{Conclusion} \label{sec:conclusion}
We have shown that for three parametric clocks with a single integer parameter, both
the reachability and safety problems are undecidable in the discrete as well as 
the continuous semantics. This improves the previously
known undecidability result by Alur, Henzinger and Vardi~\cite{AHV:STOC:93} 
where six parameters were needed. 
For the case with a single parametric clock with an unrestricted number of
integer parameters and with any number of additional nonparametric clocks, 
we contributed to the solution of an open problem stated more than 20 years ago
by proving a decidability result for reachability and safety problems
in the continuous semantics, extending the previously known 
decidability result for the discrete-time semantics~\cite{AHV:STOC:93}.
To achieve this result, we used the corner-point abstraction technique that 
had to be modified to handle also corner-points in unbounded regions,
contrary to the use of the technique in~\cite{BFHLPRV:HSCC:01}.
Not surprisingly, the decidability of the problem in case of two 
parametric clocks in the continuous-time setting remains open, 
as it is the case also for a number
of other problems over timed automata with two real-time 
clocks~\cite{BBM:IPL:06}.
On the other hand, as demonstrated by our wireless fire alarm case study,
the parameter synthesis problem for one parametric clock and an unlimited
number of parameters is sufficiently expressive in order to describe
nontrivial scheduling problems. As a next step, we will consider
moving from corner-point regions into zones and provide an
efficient implementation of the presented techniques. 

\paragraph{Acknowledgements.}
We acknowledge a funding from
from the EU FP7 grant agreement nr. 318490 (SENSATION) and grant
agreement nr. 601148 (CASSTING) and 
from the Sino-Danish Basic Research Center IDEA4CPS.

\bibliographystyle{plain}
\bibliography{main}

\begin{thebibliography}{10}

\bibitem{ACD:LICS:90}
R.~Alur, C.~Courcoubetis, and D.~Dill.
\newblock Model-checking for real-time systems.
\newblock In {\em {LICS}'90}, pages 414--425. IEEE, 1990.

\bibitem{ad:tcs94}
R.~Alur and D.~Dill.
\newblock A theory of timed automata.
\newblock {\em Theoretical Computer Science}, 126(2):183--235, 1994.

\bibitem{AHV:STOC:93}
R.~Alur, T.A. Henzinger, and M.Y. Vardi.
\newblock Parametric real-time reasoning.
\newblock In {\em Proceedings of 25th Annual Symposium on Theory of Computing
  ({STOC}'93)}, pages 592--601. ACM Press, 1993.

\bibitem{BFHLPRV:HSCC:01}
Gerd Behrmann, Ansgar Fehnker, Thomas Hune, Kim~Guldstrand Larsen, Paul
  Pettersson, Judi Romijn, and Frits~W. Vaandrager.
\newblock Minimum-cost reachability for priced timed automata.
\newblock In {\em {HSCC}'01}, volume 2034 of {\em LNCS}, pages 147--161.
  Springer, 2001.

\bibitem{BBM:IPL:06}
Patricia Bouyer, Thomas Brihaye, and Nicolas Markey.
\newblock Improved undecidability results on weighted timed automata.
\newblock {\em Inform. Proc. Letters}, 98(5):188--194, 2006.

\bibitem{BT:FMSD:09}
Laura Bozzelli and Salvatore La~Torre.
\newblock Decision problems for lower/upper bound parametric timed automata.
\newblock {\em Formal Methods in Syst. Design}, 35(2):121--151, 2009.

\bibitem{BHJL:RP:13}
Beatrice Bérard, Serge Haddad, Aleksandra Jovanović, and Didier Lime.
\newblock Parametric interrupt timed automata.
\newblock In {\em {RP}'13}, volume 8169 of {\em LNCS}, pages 59--69. Springer,
  2013.

\bibitem{BR:FSTTCS:03}
Véronique Bruyère and Jean-François Raskin.
\newblock Real-time model-checking: Parameters everywhere.
\newblock In {\em {FST TCS}'03}, volume 2914 of {\em LNCS}, pages 100--111.
  Springer, 2003.

\bibitem{BO:MFCS:14}
Daniel Bundala and Joël Ouaknine.
\newblock Advances in parametric real-time reasoning.
\newblock In {\em {MFCS}'14}, volume 8634 of {\em LNCS}, pages 123--134.
  Springer, 2014.

\bibitem{Doyen2007208}
Laurent Doyen.
\newblock Robust parametric reachability for timed automata.
\newblock {\em Information Processing Letters}, 102(5):208 -- 213, 2007.

\bibitem{alarm:2014}
Sergio Feo-Arenis, Bernd Westphal, Daniel Dietsch, Marco Muñiz, and Siyar
  Andisha.
\newblock The wireless fire alarm system: Ensuring conformance to industrial
  standards through formal verification.
\newblock In {\em FM'14}, volume 8442 of {\em LNCS}, pages 658--672. Springer,
  2014.

\bibitem{hune2001linear}
Thomas Hune, Judi Romijn, Mari{\"e}elle Stoelinga, and Frits Vaandrager.
\newblock {\em Linear parametric model checking of timed automata}.
\newblock Springer, 2001.

\bibitem{j-TACAS-13}
A.~Jovanovi\'{c}, D.~Lime, and O.~H. Roux.
\newblock Integer parameter synthesis for timed automata.
\newblock In {\em TACAS 2013}, volume 7795 of {\em LNCS}, pages 401--415.
  Springer, 2013.

\bibitem{KCH:ACCS:10}
E.V. Kuzmin and D.J. Chalyy.
\newblock Decidability of boundedness problems for {Minsky} counter machines.
\newblock {\em Automatic Control and Computer Sciences}, 44(7):387--397, 2010.

\bibitem{miller2000decidability}
Joseph~S Miller.
\newblock Decidability and complexity results for timed automata and
  semi-linear hybrid automata.
\newblock In {\em Hybrid Systems: Computation and Control}, pages 296--310.
  Springer, 2000.

\bibitem{Minsky:book}
M.L. Minsky.
\newblock {\em Computation: Finite and Infinite Machines}.
\newblock Prentice, 1967.

\bibitem{ACFE:ENTCS:08}
Étienne André, Thomas Chatain, Laurent Fribourg, and Emmanuelle Encrenaz.
\newblock An inverse method for parametric timed automata.
\newblock {\em ENTCS}, 223(0):29--46, 2008.

\bibitem{Wang1996131}
Farn Wang.
\newblock Parametric timing analysis for real-time systems.
\newblock {\em Information and Computation}, 130(2):131--150, 1996.

\end{thebibliography}

\clearpage
\appendix
\section{Appendix: Proof of Lemma~\ref{lem:minsky}}\label{app:undec}

\textbf{Lemma~\ref{lem:minsky}.}
\textit{Let $M$ be a Minsky machine.
Let $A$ be the PTA built according to the rules in Figures~\ref{fig:encoding:inc} and \ref{fig:encoding:dec} (without the 
transitions for safety) and where $\ell_1$ is the initial 
location and $\ell_n$ is the only accepting location.
The Minsky machine $M$ halts if and only if there is a parameter
valuation $\paramval$ such that $L_\paramval(A) \not= \emptyset$.}

\begin{proof}
We shall first argue that from the configuration $(\ell_i,\clockval)$
where $\clockval(z)=0$ and where $\clockval(x_1)$ 
and $\clockval(x_2)$ represent the counter values, there is a unique 
way to move from $\ell_i$ to $\ell_j$ (or possibly also to $\ell_k$ in case of 
test and decrement instruction) where again $\clockval(z)=0$ and the counter
values are updated accordingly. As there are no invariants in the automaton,
we can always delay long enough so that we get stuck in a given location,
but this behaviour will not influence the language emptiness problem we
are interested in.

Consider first the automaton for the increment instruction from
Figure~\ref{fig:encoding:inc} and assume we are in a configuration 
$(\ell_i,\clockval)$ where $\clockval(z)=0$, $\clockval(x_1)=v_1$ 
and $\clockval(x_2)=v_2$.
First note that if $v_1\geq p$ then there is no execution ending in $\ell_k$
due to the forced delay of one time unit on the transition from
$\ell_i$ to $\ell_i^1$ and the guard $x_1=p$ tested in both the upper
and lower branch in the automaton. Assume thus that $v_1 < p$.
If $v_1 \geq v_2$ then we can perform the following execution
with uniquely determined time delays:
\begin{align*}
(\ell_i, [x_1 \mapsto v_1, x_2 \mapsto v_2, z \mapsto 0])
&\goes{1}\\
(\ell_i^1, [x_1 \mapsto v_1+1, x_2 \mapsto v_2+1, z \mapsto 0])
&\goes{p-v_1-1}\\
(\ell_i^2, [x_1 \mapsto 0, x_2 \mapsto p - v_1 + v_2, z \mapsto p-v_1-1])
&\goes{v_1-v_2}\\
(\ell_i^3, [x_1 \mapsto v_1-v_2, x_2 \mapsto 0, z \mapsto p-v_2-1])
&\goes{1}\\
(\ell_i^4, [x_1 \mapsto v_1-v_2+1, x_2 \mapsto 0, z \mapsto p-v_2])
&\goes{v_2}\\
(\ell_j, [x_1 \mapsto v_1+1, x_2 \mapsto v_2, z \mapsto 0])
&.
\end{align*}
In this case where $v_1 \geq v_2$, executing the lower branch of the automaton
will result in getting stuck in the location $\ell_i^6$ as here
necessarily $\clockval(x_1)>p$.
Assume now that $v_1 < v_2$. If we take upper branch in the automaton now then
we get stuck. However, we can execute along the lower branch as follows:
\begin{align*}
(\ell_i, [x_1 \mapsto v_1, x_2 \mapsto v_2, z \mapsto 0]) 
&\goes{1}\\
(\ell_i^1, [x_1 \mapsto v_1+1, x_2 \mapsto v_2+1, z \mapsto 0]) 
&\goes{p-v_2-1}\\
(\ell_i^5, [x_1 \mapsto p- v_2 + v_1, x_2 \mapsto 0, z \mapsto p-v_2-1]) 
&\goes{1}\\
(\ell_i^6, [x_1 \mapsto p- v_2 + v_1+1, x_2 \mapsto 0, z \mapsto p-v_2]) 
&\goes{v_2 -v_1 -1}\\
 (\ell_i^4, [x_1 \mapsto 0, x_2 \mapsto  v_2 -v_1-1, z \mapsto p-v_1-1]) 
&\goes{v_1 +1}\\
 (\ell_j, [x_1 \mapsto v_1 + 1, x_2 \mapsto  v_2 , z \mapsto 0])
&. 
\end{align*}
In both cases, it is clear that the clock valuation of $x_1$ was
incremented by one (due to uniquely given time delays during the computations)
and hence they faithfully simulate the behaviour of the
Minsky machine.

Regarding the automaton for test and decrement instruction from 
Figure~\ref{fig:encoding:dec}, assume we are in the configuration 
$(\ell_i,\clockval)$ with $\clockval(x_1)=v_1$, $\clockval(x_2)=v_2$ 
and $\clockval(z)=0$.
It is clear that if $v_1=0$ then we continue
from location $\ell_k$ as expected and if $v_1>0$ then we can enter the 
configuration $(\ell_i^1, x_1 \mapsto v_1, x_2 \mapsto v_2, z \mapsto 0)$---
note that the guard $z=0$ guarantees that no time has elapsed.
In the latter case, if $v_1 > v_2$ then we can execute the  
upper branch as follows:
\begin{align*}
(\ell_i^1, x_1 \mapsto v_1, x_2 \mapsto v_2, z \mapsto 0)
&\goes{p-v_1}\\
(\ell_i^2, x_1 \mapsto 0, x_2 \mapsto p - v_1 + v_2, z \mapsto p-v_1)
&\goes{1}\\
(\ell_i^3, x_1 \mapsto 0, x_2 \mapsto p - v_1 + v_2 + 1, z \mapsto p-v_1 +1)
&\goes{v_1 - v_2 - 1}\\
(\ell_i^4, x_1 \mapsto v_1 - v_2 - 1, x_2 \mapsto 0, z \mapsto p- v_2)
&\goes{v_2}\\
(\ell_j, x_1 \mapsto v_1 - 1, x_2 \mapsto v_2, z \mapsto 0)
\end{align*}
and if $v_1 \leq v_2$ we can execute the lower branch as follows:
\begin{align*}
(\ell_i^1, x_1 \mapsto v_1, x_2 \mapsto v_2, z \mapsto 0)
&\goes{p-v_2}\\
(\ell_i^5, x_1 \mapsto p-v_2 + v_1, x_2 \mapsto 0, z \mapsto p-v_2)
&\goes{v_2 - v_1}\\
(\ell_i^6, x_1 \mapsto 0, x_2 \mapsto  v_2 -v_1, z \mapsto p-v_1)
&\goes{1}\\
(\ell_i^4, x_1 \mapsto 0, x_2 \mapsto v_2 - v_1+1, z \mapsto p- v_1 +1)
&\goes{v_1 - 1}\\
(\ell_j, x_1 \mapsto v_1 - 1, x_2 \mapsto v_2, z \mapsto 0)
&.
\end{align*}
Clearly, the clock value in $x_1$ has been decremented in both cases.
Should the lower branch be taken in case $v_1 > v_2$ or the upper branch
in case $v_1 \leq v_2$, we get stuck again.

Now if the Minsky machine halts then
in the constructed PTA we can reach the accepting location $\ell_n$ for
any parameter valuation $\paramval(p)$ larger than the maximum
value of the counters during the computation and the answer to 
the reachability problem is hence positive. If, on the other hand, 
the Minsky machine loops
then there is no parameter valuation $\paramval(p)$ that will allow us
to reach the location~$\ell_n$. This is due to the fact that either 
one of the counters exceeds
the chosen parameter value and we get stuck or the computation will
continue for ever and never reach $\ell_n$.
\qed
\end{proof}

\section{Appendix: Proof of Theorem~\ref{thm:main}}\label{app:dec}

\textbf{Theorem~\ref{thm:main}.}
\textit{The reachability and safety problems for parametric timed automata over integer
parameters with one parametric clock in the continuous-time semantics are
decidable. Moreover, the reachability problem is in NEXPTIME.}

\medskip

The semantics of a~0/1-timed automaton is similar to that of a~timed automaton
except for the fact that the delays are explicitly given by the 0/1-delay
transitions and the valuations of clocks are natural numbers.
In our case we have only one clock, $\hatparamclock$, which 
means that a~configuration of $\hat A$ is
a~tuple $(\ell,r,\alpha,t)$ where $(\ell,r,\alpha)$ is a~location of $\hat A$
and $t \in \N$ represents the valuation of $\hatparamclock$.

Note that the 0/1-delay transitions in $\hat A$ are always deterministic and exclusive:
every location $(\ell,r,\alpha)$ either has an outgoing 0-delay transition or 
an~outgoing 1-delay transition, but not both. Moreover, after a~1-delay transition
there always follows a~0-delay transition, except in the case when the 1-delay transition
ends in the unbounded region.

Recall the $\iota(r,\alpha)$ notation for $r \in \Regions$
and $\alpha \in \CP$:
\[ \iota(r,\alpha) = \begin{cases}
	\LESS & \alpha(\newclock) = 1 \text{ and } r \not\models \newclock = 1\\
	\MORE & \alpha(\newclock) = 0 \text{ and } r \not\models \newclock = 0\\
	\EXACT&\text{otherwise}
\end{cases}\]

In order to prove the correctness of our construction, we also define 
an~auxiliary notion of \emph{correspondence}.
Let $\nu : (\clocks \cup \{\newclock\}) \to \R$, $r \in \Regions$, $\alpha \in \CP$, and
$t \in \N$. We say that $\nu$ corresponds with $(r,\alpha,t)$ if the following holds:
\begin{enumerate}
\item $\restr{\nu}{\hat\clocks} \in r$;
\item $\lfloor \nu(\paramclock) \rfloor + f(r,\alpha) = t$,
where $f(r,\alpha) = 1$ if $\iota(r,\alpha) = \LESS$ and $f(r,\alpha) = 0$ otherwise;
\item $\nu(\paramclock) \in \N \iff \iota(r,\alpha) = \EXACT$.
\end{enumerate}

In the following, let us fix a valuation of all parameters of $\parameters$.
We are going to show that all runs of $A'$ have corresponding runs in $\hat A$ and
vice versa. Note that due to the construction of $A'$, we shall ignore all runs in
which the new clock~$\newclock$ becomes larger than 1, as such situations are
effectively deadlocks.

\begin{lemma}\label{lem:delay-to-01}
Let $(\ell,\nu)$ be a~configuration of $A'$ with $\nu(\newclock) \le 1$,  let
$(\ell,\nu) \goes{d} (\ell,\nu+d)$ be a~delay transition with 
$d \in [0, 1-\nu(\newclock)]$, and let
$(r,\alpha,t) \in \Regions \times \CP \times \N$ such that 
$\nu$ corresponds with $(r,\alpha,t)$. 
Then $(\ell,r,\alpha,t) \goes{}^* (\ell,r',\alpha',t')$ such that
$\nu+d$ corresponds with $(r',\alpha',t')$.
\end{lemma}
\begin{proof}
For simplicity, we assume that $d$ is small enough in the following sense:
Either $\nu$ and $\nu + d$ (restricted to clocks of $\hat\clocks$) are
in the same region, or the region of $\nu + d$ is a~successor of the region of $\nu$.
This comes without loss of generality, as every delay transition can be split
into finitely many such small delay transitions.
We also assume that $d > 0$.

If both $\restr{\nu}{\hat\clocks}$ and $\restr{(\nu + d)}{\hat\clocks}$
belong to $r$ then clearly 
$\lfloor\nu(\paramclock)+d\rfloor = \lfloor\nu(\paramclock)\rfloor$ and
neither of $\nu(\paramclock)$ and $\nu(\paramclock) + d$ belong to $\N$.
This means that $\nu + d$ corresponds with $(r,\alpha,t)$.

Let us now assume that $\restr{\nu}{\hat\clocks} \in r$ and
$\restr{(\nu + d)}{\hat\clocks} \in r'$ where $r'$ is the successor of $r$.
There are two possibilities depending on whether $(\ell,r,\alpha)$ has 
a~0-delay or a~1-delay transition.

If $(\ell,r,\alpha) \goes{0} (\ell,r',\alpha)$, this means that also
$(\ell,r,\alpha,t) \goes{0} (\ell,r',\alpha,t)$. We show that $\nu + d$
corresponds with $(r',\alpha,t)$. Condition 1 is clearly satisfied.
To show the satisfaction of Conditions 2 and 3, we need to discuss three cases:
\begin{itemize}
\item $\nu(\newclock) = 0$: This means that necessarily $r \models \newclock = 0$ and
$\alpha(\newclock) = 0$, thus $\iota(r,\alpha) = \EXACT$ and $\iota(r',\alpha) = \MORE$.
Therefore $f(r,\alpha) = f(r',\alpha) = 0$ and the conditions are clearly met
as $\lfloor\nu(\paramclock)+d\rfloor = \lfloor\nu(\paramclock)\rfloor$
and $\nu(\paramclock) + d \not \in \N$.
\item $\nu(\newclock) \in (0,1-d)$: This means that necessarily
neither $r$ nor $r'$ contain a~valuation with integer value for $\newclock$
and thus $\iota(r,\alpha) = \iota(r',\alpha) \ne \EXACT$. This means that
$f(r,\alpha) = f(r',\alpha)$ and the conditions are clearly met
as $\lfloor\nu(\paramclock)+d\rfloor = \lfloor\nu(\paramclock)\rfloor$
and both $\nu(\paramclock)$, $\nu(\paramclock) + d \not \in \N$.
\item $\nu(\newclock) = 1-d$: This means that necessarily
$r' \models \newclock = 1$ and $\alpha(\newclock) = 1$,
thus $\iota(r,\alpha) = \LESS$ and $\iota(r',\alpha) = \EXACT$.
This means that $f(r,\alpha) = 1$ while $f(r',\alpha) = 0$.
We have $\lfloor\nu(\paramclock)+d\rfloor = \lfloor\nu(\paramclock)\rfloor+1$
and $\nu(\paramclock)+d \in \N$, the conditions are thus met again.
\end{itemize}

If $(\ell,r,\alpha) \goes{1} (\ell,r,\alpha')$ then also
$(\ell,r,\alpha') \goes{0} (\ell,r',\alpha')$ as noted above (due to the bound on
$\nu(z)$, $r$ is not the unbounded region).
This means that $(\ell,r,\alpha,t) \goes{1}\goes{0} (\ell,r',\alpha',t+1)$.
Here, $\alpha'$ is the successor of $\alpha$ and $r'$ is the successor of $r$.
We show that $\nu + d$ corresponds with $(r',\alpha',t+1)$.
Again, Condition 1 is clearly satisfied.
Note that in this case $r \not\models \newclock = 0$ and $\alpha(\newclock) = 0$.
This means that $\alpha'(\newclock) = 1$, $\iota(r,\alpha) = \MORE$,
$\iota(r,\alpha') = \LESS$, $f(r,\alpha) = 0$, and $f(r,\alpha') = 1$.
This also means that $\nu(\newclock) > 0$ and we only have two cases:
\begin{itemize}
\item $\nu(\newclock) \in (0,1-d)$:
This means that $r' \not\models \newclock = 1$ and thus $\iota(r',\alpha') = \LESS$.
Condition 3 is clearly satisfied as $\nu(\paramclock)+d \not\in \N$.
To show that Condition 2 is satisfied, consider that $f(r',\alpha') = 1$
and $\lfloor\nu(\paramclock)+d\rfloor = \lfloor\nu(\paramclock)\rfloor$.
\item $\nu(\newclock) = 1-d$:
This means that $r' \models \newclock = 1$ and thus $\iota(r',\alpha') = \EXACT$.
Condition 3 is clearly satisfied as $\nu(\paramclock)+d \in \N$.
To show that Condition 2 is satisfied, consider that $f(r',\alpha') = 0$
and
$\lfloor\nu(\paramclock)+d\rfloor = \lfloor\nu(\paramclock)\rfloor + 1 = t + 1$.
\end{itemize}
\qed
\end{proof}

\begin{lemma}\label{lem:act-to-01}
Let $(\ell,\nu)$ be a~configuration of $A'$ with $\nu(\newclock) \le 1$,  let
$(\ell,\nu) \goes{a} (\ell',\nu')$ be an~action transition, and let
$(r,\alpha,t) \in \Regions \times \CP \times \N$ such that 
$\nu$ corresponds with $(r,\alpha,t)$. 
Then $(\ell,r,\alpha,t) \goes{a} (\ell',r',\alpha',t')$ such that
$\nu'$ corresponds with $(r',\alpha',t')$.
\end{lemma}
\begin{proof}
The transition $(\ell,\nu) \goes{a} (\ell',\nu')$ is due to a~transition
$\ell \goes{g,a,R} \ell'$ of the timed automaton $A'$ where
$\nu \models g$ and $\nu' = \nu[R]$.
Let now $g_1$, \ldots, $g_k$ be all the simple clock constraints of $g$
that consider clocks from $\hat\clocks$ and $h_1$, \ldots, $h_n$ be
the remaining simple clock constraints, as in the construction of the
0/1-timed automaton $\hat A$.
We know that $r \models g_1 \wedge \cdots g_k$ as 
$\restr{\nu}{\hat\clocks} \in r$.
We thus know that $(\ell,r,\alpha) \goes{\hat{h}_1\wedge\cdots\hat{h}_n,a,
\hat R} (\ell',r[R\setminus\{\paramclock\}],
\alpha[R\setminus\{\paramclock\}]) = (\ell',r',\alpha')$,
where $\hat{h}_i$ and $\hat R$ are given by the construction.
We first need to show that $t$ satisfies all clock constraints $\hat{h}_i$.
\begin{itemize}
\item If $h_i = (\paramclock < e)$ then $\nu(\paramclock) < e$ and thus
$\lfloor \nu(\paramclock) \rfloor \le e - 1$. If $\iota(r,\alpha) = \LESS$
then $f(r,\alpha) = 1$ and $t \le (e-1) + 1 = e$, which satisfies 
the constraint $\hat{h}_i = (\hatparamclock \le e)$. Otherwise,
$f(r,\alpha) = 0$ and $t \le e-1$, which satisfies the constraint 
$\hat{h}_i = (\hatparamclock < e)$.
\item If $h_i = (\paramclock \le e)$ then $\nu(\paramclock) \le e$.
If $\nu(\paramclock) = e$ then $\iota(r,\alpha) = \EXACT$. Thus 
$f(r,\alpha)=0$ and $t = e$, which satisfies $\hatparamclock \le e$.
Otherwise, $\nu(\paramclock) < e$. Due to the reasoning in the previous item,
$t \le e$ if $\iota(r,\alpha)=\LESS$ and $t \le e - 1$ otherwise.
This means that $t$ satisfies $\hat h_i$.
\item The other two cases ($>$, $\ge$) are dealt with similarly.
\end{itemize}
We have thus shown that the transition $(\ell,r,\alpha,t) \goes{a}
(\ell',r',\alpha',t')$ exists. Here, $t' = 0$ if $\paramclock \in R$ and
$t' = t$ otherwise.

We now show that $\nu'$ corresponds with $(r',\alpha',t')$. Condition 1
is clearly satisfied. If $t' = t$ then both $\paramclock, \newclock \not\in R$, 
which means that both $\nu(\paramclock)$ and $\iota(r,\alpha)$ are unchanged.
If $t' = 0$ then $\paramclock, \newclock \in R$, which means that
$\nu(\paramclock) = 0$, $\iota(r',\alpha') = \EXACT$ and $f(r',\alpha') = 0$.
In both cases, $\nu'$ corresponds with $(r',\alpha',t')$.
\qed
\end{proof}

\begin{lemma}\label{lem:01-to-delay}
Let $(\ell,r,\alpha,t)$ be a~configuration of $\hat A$,
let $(\ell,r,\alpha,t) \goes{d} (\ell,r',\alpha',t')$ with $d \in \{0,1\}$
and $r' \models \newclock \le 1$,
and let $\nu$ correspond with $(r,\alpha,t)$.
Then there exists $d'$ such that $(\ell,\nu) \goes{d'} (\ell,\nu+d')$
and $\nu+d'$ corresponds with $(r',\alpha',t')$.
\end{lemma}
\begin{proof}
If $d = 1$, we choose $d' = 0$ and show that $\nu + 0 = \nu$ corresponds
with $(r',\alpha',t+1)$.
Clearly, in this case $r' = r$, $\alpha(\newclock) = 0$, and 
$\alpha'(\newclock) = 1$,
which means that $\iota(r,\alpha) = \MORE$ while $\iota(r,\alpha') = \LESS$.
Therefore, $f(r,\alpha) = 0$, $f(r,\alpha') = 1$ and
if $\lfloor \nu(\paramclock) \rfloor = t$ then 
$\lfloor \nu(\paramclock) \rfloor + 1 = t + 1$.

Let us now assume that $d = 0$. This means that $\alpha' = \alpha$ while
$r'$ is the successor region of $r$. We choose an arbitrary $d'$ such that
$\restr{(\nu + d')}{\hat\clocks} \in r'$. To show that $\nu + d'$ corresponds
with $(r',\alpha,t)$, we can use the very same reasoning as in the proof of 
Lemma~\ref{lem:delay-to-01}.
\qed
\end{proof}

\begin{lemma}\label{lem:01-to-act}
Let $(\ell,r,\alpha,t)$ be a~configuration of $\hat A$
with $r \models \newclock \le 1$,
let $(\ell,r,\alpha,t) \goes{a} (\ell',r',\alpha',t')$,
and let $\nu$ correspond with $(r,\alpha,t)$.
Then $(\ell,\nu) \goes{a} (\ell',\nu')$
and $\nu'$ corresponds with $(r',\alpha',t')$.
\end{lemma}
\begin{proof}
The transition $(\ell,r,\alpha,t) \goes{a} (\ell',r',\alpha',t')$ 
in the semantics
is due to a~transition $(\ell,r,\alpha) \goes{\hat{h}_1\wedge\cdots\hat{h}_n,
a,\hat R} (\ell',r',\alpha')$ of the 0/1-timed automaton $\hat A$,
which was constructed from a~transition $\ell \goes{g,a,R} \ell'$ of the
timed automaton $A'$.

Clearly, if $r \models g_1 \wedge \cdots g_k$ then so does $\nu$.
We also need to show that $\nu \models h_1 \wedge \cdots h_n$.
We know that $t$ satisfies $\hat{h}_i$ for all $i$.
\begin{itemize}
\item If $h_i = (\paramclock < e)$ then either $\iota(r,\alpha) = \LESS$
and $\hat{h}_i = (\hatparamclock \le e)$ or $\hat{h}_i = (\hatparamclock < e)$.
In the first case, we know that $f(\alpha,r) = 1$, which means that 
$\lfloor \nu(\paramclock) \rfloor + 1 = t$, which implies that 
$\nu(\paramclock) < t \le e$. In the second case, 
$f(\alpha,r) = 0$ which means that $\lfloor\nu(\paramclock) \rfloor = t$
and thus $\nu(\paramclock) < t + 1 \le e$ as $t \le e - 1$.
\item If $h_i = (\paramclock \le e)$ then either $\iota(r,\alpha) = \MORE$
and $\hat{h}_i = (\hatparamclock < e)$ or $\hat{h}_i = (\hatparamclock \le e)$.
In both cases, if $t < e$ then $t \le e - 1$ and thus
$\lfloor \nu(\paramclock)\rfloor + f(r,\alpha) \le e-1$. This means that no
matter the value of $f(r,\alpha)$, $\nu(\paramclock) < e$.
If, in the second case, $t = e$ then 
$\lfloor \nu(\paramclock)\rfloor + f(r,\alpha) = e$. If $f(r,\alpha) = 1$
then $\nu(\paramclock) < e$. If $f(r,\alpha) = 0$ then this means that
$\iota(r,\alpha) = \EXACT$ and $\nu(\paramclock) = e$ as 
$\nu(\paramclock) \in \N$.
\item The remaining two cases ($>$, $\ge$) are dealt with similarly.
\end{itemize}

We now need to show that $\nu' = \nu[R]$ corresponds with $(r',\alpha',t')$.
This is shown exactly as in the proof of Lemma~\ref{lem:act-to-01}.
\qed
\end{proof}

The correctness of the construction is now a~corollary of the previous four
lemmata; this proves the main theorem.

\end{document}